\documentclass[runningheads]{llncs}
\usepackage{amsmath,amssymb}
\usepackage{algorithm,algpseudocode}
\usepackage{stmaryrd}
\usepackage{forest}
\usepackage{tikz}

\newcommand{\N}{\mathbb{N}}
\newcommand{\id}{\mathit{id}}
\newcommand{\D}{\mathbb{D}}
\newcommand{\A}{\mathcal{A}}
\renewcommand{\S}{\mathcal{S}}
\newcommand{\E}{\mathcal{E}}
\newcommand{\T}{\mathcal{T}}
\newcommand{\Orbit}{\mathit{Orbit}}
\newcommand{\Sym}{\mathsf{Sym}}
\newcommand{\Pos}{\mathit{Pos}}
\newcommand{\arity}{\mathit{arity}}
\newcommand{\Tree}{\mathit{Tree}}
\newcommand{\Subtree}{\mathit{Subtree}}
\newcommand{\Context}{\mathit{Context}}
\newcommand{\Next}{\mathit{Next}}
\newcommand{\row}{\mathit{row}}
\newcommand{\lec}{\llbracket}
\newcommand{\rec}{\rrbracket}
\newcommand{\CSec}{{\lec C,S \rec}}
\newcommand{\DTec}{{\lec D,T \rec}}
\newcommand{\Q}{\tilde{Q}}
\newcommand{\F}{\tilde{F}}
\newcommand{\del}{\tilde{\delta}}

\newcommand{\rSET}{\row}
\newcommand{\rSETd}{\row'}
\newcommand{\trm}{\mathit{trm}}
\newcommand{\accept}{\mathit{accept}}
\newcommand{\reject}{\mathit{reject}}

\title{Active Learning for Deterministic Bottom-up Nominal Tree Automata}
\author{Rindo Nakanishi\inst{1} \and
	Yoshiaki Takata\inst{2} \and
	Hiroyuki Seki\inst{1}}
\authorrunning{R. Nakanishi et al.}
\institute{Graduate School of Informatics, Nagoya University \\
	Furo-cho, Chikusa, Nagoya 464-8601, Japan \\
	\email{rindo@sqlab.jp}, \email{seki@i.nagoya-u.ac.jp}
		\and Graduate School of Engineering, Kochi University of Technology \\
	Tosayamada, Kami City, Kochi 782-8502, Japan \\
	\email{takata.yoshiaki@kochi-tech.ac.jp}}

\begin{document}
\maketitle
\begin{abstract}
	Nominal set plays a central role in a group-theoretic extension of 
finite automata to those over an infinite set of data values. 
Moerman et al. proposed an active learning algorithm for
nominal word automata with the equality symmetry.
In this paper, we introduce deterministic bottom-up nominal tree automata (DBNTA), 
which operate on trees whose nodes are labelled with 
elements of an orbit finite nominal set. 
We then prove a Myhill-Nerode theorem for the class of languages recognized by DBNTA
and propose an active learning algorithm for DBNTA. 
The algorithm can deal with any data symmetry that admits least support, 
not restricted to the equality symmetry and/or the total order symmetry. 
To prove the termination of the algorithm, 
we define a partial order on nominal sets and show that 
there is no infinite chain of orbit finite nominal sets 
with respect to this partial order between any two orbit finite sets. 

\end{abstract}

\section{Introduction}
	Computational models such as finite automaton, pushdown automaton and context-free grammar
provide a theoretical basis for automated technologies including model checking, testing and synthesis. 
Although the technologies have brought fruitful success, these models cannot directly deal with data values. 
However, if we add to a classical model the ability of processing data values, the resulting model easily 
becomes Turing machine-equivalent and the decidability needed for automated technologies is lost. 
Register automaton (RA) is an extension of finite automaton (FA) by adding registers for manipulating 
data values in a restricted way \cite{KF94}. 
RA can compare an input data value with those stored in its registers to determine its behavior. 
RA inherits some of the good properties from FA including closure properties on language operations
and the decidability of basic problems. 
For example, the membership and emptiness problems are decidable for RA
and their complexities are extensively studied \cite{KF94,SI00,NSV04,LTV15}. 
Similar extensions of other classical models have been proposed 
such as register tree automaton \cite{KT00,FS11}, register context-free grammar \cite{CK98,STS18} and register pushdown automaton 
\cite{MRT17,STS21}. 
Logics on data words have also been proposed including 
LTL with the freeze quantifier \cite{DL09} and two-variable first-order logic \cite{BDMSS11}. 

A common property of these models is that the behavior of an automaton (or a grammar)
does not depend on data values themselves, but on the relationship (e.g., equality, total order) among data values. 
Assume that the comparison operator of RA is only equality check.
Also assume that an RA $A$ has one register to store the first data value of an input word and 
test whether the remaining data values are different from the data value in the register 
except the last one in the input, which should be the same as the first data value. 
Then, $A$ accepts data words $2\cdot 5\cdot 6\cdot 2$, $8\cdot 1\cdot 3\cdot 8$, and so on. 
Note that the data word $8\cdot 1\cdot 3\cdot 8$ can be obtained from $2\cdot 5\cdot 6\cdot 2$ 
by the permutation that maps $2, 5$ and $6$ to $8, 1$ and $3$, respectively. 

The above observation gives us a group-theoretic extension of FA \cite{BKL14}. 
Assume that we are given a countable set ${\mathbb D}$ of data values and a permutation group $G$ on ${\mathbb D}$. 
To deal with data values in a restricted way, 
we use orbit finite sets instead of finite sets to represent both of an alphabet and a set of states. 
A set $X$ is called a $G$-set if $X$ is equipped with actions (or operations) having the group structure $G$. 
Let $X$ be a $G$-set. 
The orbit of $x\in X$ is the set $\{ x\cdot \pi \mid  \pi \in G  \}$. 
$X$ is {\em orbit finite} if $X$ is divided into a finite number of orbits. 
A set $C\subseteq {\mathbb D}$ is a {\em support} of $x\in X$ 
if any action $\pi$ that acts as identity on $C$ does not move $x$, i.e., maps $x$ to $x$. 
$X$ is {\em nominal} if every $x\in X$ has a finite support. 
Intuitively, $X$ is nominal if for every $x\in X$, all the information on $x$
can be represented by a finite subset of data values, which corresponds to the contents of registers. 
A nominal automaton over an orbit finite alphabet consists of 
an orbit finite nominal set of states and an (equivariant) transition relation on states. 

Automated learning methods are incorporated into software verification and testing 
(see \cite{Le06,CGKPV18} for an overview). 
Two well-known applications are black-box checking \cite{PVY99} and compositional verification \cite{CGP03}. 
Among others, Angluin's $L^{\ast}$ algorithm \cite{An87} is frequently used in these methods. 
The algorithm learns the minimum FA for an unknown regular language $U$
by constructing an observation table. 
Rows and columns of the table are sample input strings and each entry of the table is 1 (accept) or 0 (reject). 
The algorithm expands the table based on answers from a teacher (oracle) of $U$ for membership and equivalence queries
until the teacher answers yes to an equivalence query. 
The correctness of $L^{\ast}$ is guaranteed by the Myhill-Nerode theorem for regular languages. 
The $L^{\ast}$ algorithm has been extended for register automata (e.g. \cite{BHLM13,CHJS16})
and a learning tool RALib is implemented \cite{CHJS16}. 
In RA, a state transition depends on the comparison among an input data value and those stored in the registers
specified as the guard condition of an applied transition rule. 
Due to this feature of RA, an entry of an observation table is not just 0/1 but more complex information
that represents the guard condition of a transition in RA (a symbolic decision tree in \cite{CHJS16}), 
which makes the algorithm rather complicated. 
Moerman et al. proposed an $L^{\ast}$-style algorithm for nominal word automata \cite{MSSKS17}. 
Their algorithm recovers the simplicity of the original $L^{\ast}$ algorithm 
by the abstract feature of nominal automaton, which is independent of a concrete representation of an automaton. 
However, the algorithm assumes the equality symmetry as the structure of the set of data values. 
Moreover, tree models that can manipulate data values are needed for the basis of 
XML document processing because an XML document usually contains data values associated with structural information
represented by a tree \cite{LV12,LTV15}.  
For such applications, tree automata theory based on nominal sets should be developed. 

In this paper, we define deterministic bottom-up nominal tree automata (DBNTA), 
which operate on trees whose nodes are labelled with elements of an orbit finite nominal set. 
We then prove a Myhill-Nerode theorem for the class of languages recognized by DBNTA 
and propose an active learning algorithm for DBNTA based on the theorem. 
The algorithm can deal with any data symmetry that admits least support, 
not restricted to the equality symmetry and/or the total order symmetry. 
To prove the termination of the algorithm, 
we define a partial order on nominal sets and show that 
there is no infinite chain of orbit finite nominal sets with respect to this partial order 
between any two orbit finite sets.

\section{Preliminaries}
	\subsection{Nominal set}
Let $G$ be a group and $X$ be a set.
A group action of $G$ on $X$ is a function $\:{\cdot}:X \times G \to X$ satisfying
\[
	x \cdot e = x \qquad \text{and} \qquad  
	x \cdot (\pi \sigma) = (x \cdot \pi) \cdot \sigma
\]
for all $x \in X$ and $\pi,\sigma \in G$,
where $e \in G$ is the neutral element of $G$ 
and $\pi \sigma$ is the product of $\pi$ and $\sigma$ on $G$.
We call a set with a group action of $G$ a $G$-set.

We define the \emph{orbit} of $x \in X$ as
$\Orbit(x) = \{ x \cdot \pi \mid \pi \in G \} \subseteq X$.
A $G$-set is uniquely partitioned into different orbits.
A $G$-set consisting of one orbit is called a \emph{single orbit} set,
and a $G$-set consisting of a finite number of orbits is called an \emph{orbit finite} set.
We define an \emph{alphabet} as an orbit finite set.

Let $X$ be a $G$-set.
$Y \subseteq X$ is called \emph{equivariant} 
if $y \in Y \Rightarrow y \cdot \pi \in Y$ holds for all $\pi \in G$.
Equivalently, this means that $Y$ is a union of some orbits of $X$.
In the same way, for $G$-sets $X$ and $Y$, 
a binary relation $R \subseteq X \times Y$ is called equivariant 
if $(x,y) \in R \Rightarrow (x \cdot \pi,y \cdot \pi) \in R$ holds for all $ \pi \in G$.
An $n$-ary equivariant relation is defined in the same way for $n \ge 3$.
If a binary relation $f \subseteq X \times Y$ is a function, 
$f$ is equivariant if and only if
$f(x \cdot \pi) = f(x) \cdot \pi$ holds for all $x \in X$ and $\pi \in G$.

Let $\D$ be a countable set of data values and $G$ be a permutation group of $\D$, 
i.e., a subgroup of the symmetric group $\Sym(\D)$ of $\D$.
We call $(\D,G)$ a \emph{data symmetry}.
We show some examples of data symmetries.
The equality symmetry is $({\mathbb N}, \Sym({\mathbb N}))$, where ${\mathbb N}$ is the set of natural numbers
and $\Sym({\mathbb N})$ is the group of all bijections on ${\mathbb N}$.
The total order symmetry is
$({\mathbb Q}, G_<)$ where ${\mathbb Q}$ is the set of rational numbers and
$G_<$ is the group of monotone bijections on ${\mathbb Q}$.
The integer symmetry is
$({\mathbb Z}, G_{{\mathbb Z}})$ where ${\mathbb Z}$ is the set of integers and
$G_{{\mathbb Z}}$ is the group of translations $i \mapsto i+c$ for $c\in {\mathbb Z}$.

Let $x \in X$ and $C \subseteq \mathbb{D}$.
If for every $\pi \in G$,
\[
	(\forall c \in C.\  \pi(c) = c) \Rightarrow x \cdot \pi = x
\]
holds, we say that $C$ \emph{supports} $x$ or $C$ is a \emph{support} of $x$.
That is, $C$ supports $x$ if every $\pi$ which is the identity on $C$ does not move $x$.
A $G$-set is \emph{nominal}, if every element of the set has a finite support.
In any data symmetry $({\mathbb D}, G)$,
${\mathbb D}$ itself is a nominal $G$-set because any element $d\in{\mathbb D}$ has
a support $\{d\} \subseteq {\mathbb D}$.
${\mathbb D}^{\ast}$ is also nominal because any element $d_1d_2\cdots d_n \in{\mathbb D}^{\ast}$
has a support $\{d_1, d_2, \ldots, d_n\}$.
On the other hand, ${\mathbb D}^{\omega}$, the set of infinite sequences over ${\mathbb D}$,
is not nominal.
In the following, we just call an alphabet that is nominal a {\em nominal alphabet}.

Let $C \subseteq \D$ be a support of $x \in X$.
If all supports of $x$ are supersets of $C$, $C$ is the \emph{least support} of $x$.
For a data symmetry $(\D,G)$, if every element of every nominal $G$-set has a least support, 
the data symmetry \emph{admits least support}. 

It is shown in \cite{GP02} that
the equality symmetry and the total order symmetry admit least supports.
(Also see \cite[Corollaries 9.4 and 9.5]{BKL14}.)
In the integer symmetry, every $G_{{\mathbb Z}}$-set is nominal by the following reason.
If a translation $i \mapsto i+c$ on ${\mathbb Z}$
does not move an integer $z\in {\mathbb Z}$, the translation must be the identity.
Hence, any element $x$ of any $G_{{\mathbb Z}}$-set is supported by a singleton set of an arbitrary integer.
For the same reason, the integer symmetry does not admit least support.

For a function $f:X \to Y$ and a subset $Z \subseteq X$,
we write the function whose domain is restricted to $Z$ as $f|_Z$.
For functions $f:X \to Y$ and $g:Y' \to Z$,
we define $fg:X' \to Z$ as $fg(x) = g(f(x))$ where
$X' = \{ x \in X \mid f(x) \in Y' \}$.

Let $(\D,G)$ be a data symmetry that admits least support, $C \subseteq \D$ be a finite and 
fungible\footnote{
A finite set $C \subseteq \D$ is \emph{fungible} if for every $c \in C$ 
there exists a $\pi \in G$ such that 
$\pi(c) \neq c$ and $\pi(c') = c'$ for all $c' \in C \setminus \{c\}$.
Fungibility is a technical condition guaranteeing that
$\CSec$ is a single orbit nominal set, but
it is not directly related to the paper.
}
set and
$S \leq \Sym(C)$ be a permutation group on $C$.
For injective functions $u$ and $v$ from $C$ to $\D$ that extend to permutations from $G$
(i.e., $u = \pi|_C$ and $v = \sigma|_C$ for some $\pi,\sigma \in G$),
we define $u \equiv_S v$ if and only if $uv^{-1} \in S$
(which equivalently means $\tau u = v$ for some $\tau \in S$).
It is easy to see that $\equiv_S$ is an equivalent relation,
and thus $\equiv_S$ divides the set of all injections from $C$ to $\D$
that extend to permutations from $G$ into equivalent classes.
The equivalent class of $u$ defined by $\equiv_S$ is written as $[u]_S$.
For these $C$ and $S$, the $G$-set $\CSec$ is defined as
the set of all equivalent classes defined by $\equiv_S$,
i.e. $\CSec = \{ [\pi|_C]_S \mid \pi \in G \}$,
where the $G$-action is defined as $[u]_S \cdot \pi = [u \pi]_S$ for all $\pi \in G$.
$S$ is called a local symmetry.
As we noted before, $C$ corresponds to the set of (canonical) data values in the registers.
By definition, an element of $\CSec$ is an equivalent class defined by $\equiv_S$ of an injection from $C$ to $\D$ that extends
to some permutation in $G$.
As shown in the example in the introduction, an automaton cannot distinguish between $C$ and
the set of data values $C'$ obtained from $C$ by any injection from $C$ to $\D$
which is consistent with $G$.
Such an injection $u: C \to \D$ represents
this change of data values in the registers from $C$ to $C'$, which is indistinguishable
from the automaton.

Next, we describe an intuitive meaning of $S$.
For example, if $S$ consists of only the identity, it means that the order of registers is relevant.
If $C=\{1,2,3\}$ and $S=\{ \id, a \}$ where $\id$ is the identity and $a$ swaps $1$ and $2$,
then $S$ means that the order between the first and second values are irrelevant.
Note that a standard register automaton corresponds to $S = \{ \mathit{id} \}$.

The following two propositions guarantee that
a single orbit nominal set and an equivariant function between them have 
finite representations.

\begin{proposition}[{\cite[Proposition 9.15]{BKL14}}]\label{prop:CSec}
	\begin{enumerate}
		\item $\CSec$ is a single orbit nominal set.
		\item Every single orbit nominal set is isomorphic to some $\CSec$.	
	\end{enumerate}
\end{proposition}

$\CSec$ is called a \emph{support representation} of a single orbit nominal set.
The following proposition can be shown by \cite[Proposition 9.16]{BKL14}.
\begin{proposition}\label{prop:func}
	Let $X = \CSec$ and $Y = \DTec$ be single orbit nominal sets.
	For every equivariant function $f:X \to Y$,
	there is an injection $u$ from $D$ to $C$ satisfying $uS \subseteq Tu$ and
	$f([\pi|_C]_S) = [u]_T \cdot \pi$, for all $\pi \in G$,
	where $uS = \{ us \mid s \in S \}$ and $Tu = \{ tu \mid t \in T \}$.
	Conversely, for every injection $u:D \to C$ satisfying $uS \subseteq Tu$,
	$f([\pi|_C]_S) = [u]_T \cdot \pi$ is an equivariant function from $X$ to $Y$.
\end{proposition}

A proof of this proposition is given in the Appendix.

By Proposition~\ref{prop:func}, we can obtain a necessary and sufficient condition for
two single orbit nominal sets to be isomorphic in terms of a bijection between supports.
\begin{lemma}\label{lem:isomorphism}
	Single orbit nominal sets $\CSec$ and $\DTec$ are isomorphic if and only if
	there exists a bijection $u:D \to C$ satisfying $uS = Tu$ 
	that extends to a permutation from $G$.
\end{lemma}

\begin{proof}
	Assume that there exists a bijection $u:D \to C$ satisfying $uS = Tu$ 
	that extends to a permutation from $G$.
	Then, we have $uS \subseteq Tu$ and $u^{-1}T \subseteq Su^{-1}$.
	By Proposition \ref{prop:func}, we have two functions
	$f:\CSec \to \DTec$ and $g:\DTec \to \CSec$ such that
	\begin{align*}
		f([\pi|_C]_S) &= [u]_T \cdot \pi \\
		g([\sigma|_D]_T) &= [u^{-1}]_S \cdot \sigma.
	\end{align*}
	We show that $g$ is the inverse of $f$.
	From the assumption on $u$, there is some $\rho \in G$ satisfying $\rho|_D = u$
	(and $\rho^{-1}|_C = u^{-1}$).
	Thus, $f([\pi|_C]_S) = [u]_T \cdot \pi = [\rho|_D]_T \cdot \pi = [(\rho \pi)|_D]_T$.
	Substituting this into the definition of $g$ yields
	$g([(\rho\pi)|_D]_T) = [u^{-1}]_S \cdot (\rho \pi) = [\rho^{-1}|_C]_S \cdot (\rho\pi)
	= [(\rho^{-1}\rho\pi)|_C]_S = [\pi|_C]_S$.
	We have $g(f([\pi|_C]_S)) = [\pi|_C]_S$, and thus $g$ is the inverse of $f$.
	Therefore, $\CSec$ and $\DTec$ are isomorphic.

	Conversely, assume that $\CSec$ and $\DTec$ are isomorphic, i.e.,
	there exists some equivariant bijection $f$ from $\CSec$ to $\DTec$.
	By Proposition \ref{prop:func}, $f$ can be written as 
	$f([\pi|_C]_S) = [\sigma|_D]_T \cdot \pi$ for some $\sigma \in G$,
	such that there is an injection $u:D \to C$ satisfying 
	$\sigma|_D = u$ and $uS \subseteq Tu$.
	Also by Proposition \ref{prop:func}, $f^{-1}$ can be written as
	$f^{-1}([\pi|_D]_T) = [\rho|_C]_S \cdot \pi$ for some $\rho \in G$,
	such that there is an injection $v:C \to D$ satisfying
	$\rho|_C = v$ and $vT \subseteq Sv$.
	From this, we can derive
	$f^{-1}([\sigma|_D]_T \cdot \pi) = f^{-1}([(\sigma \pi)|_D])
		= [\rho|_C]_S \cdot (\sigma \pi)$.
	By $f([\pi|_C]_S) = [\sigma|_D]_T \cdot \pi$,
	we have $[\rho|_C]_S \cdot (\sigma \pi) = [\pi|_C]_S$, and hence
	$[\rho|_C]_S = [\pi|_C]_S \cdot (\sigma \pi)^{-1}
		= [(\sigma|_C)^{-1}]_S$.
	This means that $[v]_S = [u^{-1}]_S$.
	Thus, we have $vu \in S$.
	In the same way, we have $uv \in T$.
	By acting $u$ on $vT \subseteq Sv$, we have $uvTu \subseteq uSvu$, and hence
	we have $Tu \subseteq uS$ by $vu \in S$ and $uv \in T$.
	By $uS \subseteq Tu$ and $Tu \subseteq uS$, $uS = Tu$.
	\qed
\end{proof}

Let $X$ and $Y$ be nominal sets.
We define $Y \preceq X$ if and only if
there is an equivariant surjection from a subset of $X$ to $Y$.
If $Y \preceq X$ and $X$ and $Y$ are not isomorphic, $Y \prec X$.
We show that there is no infinite chain between two orbit finite nominal sets
$X$ and $Y$ such that $Y \preceq X$.
This property is used for proving the termination of the proposed learning algorithm.
We start with $X$ and $Y$ being single orbits.

\begin{lemma}\label{lem:finiteness}
	Let $(\D,G)$ be a data symmetry that admits least support.
	Then, for any two single orbit nominal sets $X$ and $Y$ such that $Y \prec X$,
	the length of any sequence of single orbit nominal sets $X_1,X_2,\ldots$ satisfying
	\[
		Y \prec X_1 \prec X_2 \prec \cdots \prec X
	\]
	is finite.
\end{lemma}

\begin{proof}
	By Proposition \ref{prop:CSec}, it suffices to show the lemma for
	$X = \CSec$ and $Y = \DTec$.
	Assume $\DTec \prec \CSec$.
	By definition of $\prec$, there exists an equivariant surjective function
	from a subset of $\CSec$ to $\DTec$.
	The domain of this function is $\CSec$ because $\CSec$ and $\DTec$ are single orbit sets.
	Thus, by Proposition \ref{prop:func}, there exists an injection $u:D \to C$
	satisfying $uS \subseteq Tu$ that extends to a permutation from $G$.
	Because $u$ is an injection, $|D| \leq |C|$ holds 
	where $|D|$ is the number of elements of $D$.
	If $|D| < |C|$, no injections from $D$ to $C$ are bijections, and hence
	$\CSec$ and $\DTec$ are not isomorphic by Lemma \ref{lem:isomorphism}.
	If $|D| = |C|$, then $uS \subsetneq Tu$ must hold 
	because $\CSec$ and $\DTec$ are not isomorphic.
	Thus, we have $S < u^{-1}Tu$.
	This means that $S$ is a proper subgroup of $u^{-1}Tu$ and thus $|S| < |T|$.
	Therefore, because $C$ and $D$ are finite sets and $S$ and $T$ are finite groups,
	the length of any sequence of single orbit nominal sets 
	$X_1,X_2,\ldots$ satisfying
	\[
		Y \prec X_1 \prec X_2 \prec \cdots \prec X
	\]
	is finite.
	\qed
\end{proof}

\begin{lemma}\label{lem:finiteness-f}
	Let $(\D,G)$ be a data symmetry that admits least support.
	Then, for any two orbit finite nominal sets $X$ and $Y$,
	the length of any sequence of orbit finite nominal sets $X_1,X_2,\ldots$ satisfying
	\[
		Y \prec X_1 \prec X_2 \prec \cdots \prec X
	\]
	is finite.
\end{lemma}

\begin{proof}
	Any orbit finite nominal set is an union of a finite number of single orbit nominal sets.
	Hence, this lemma obviously holds by Lemma \ref{lem:finiteness}.
	\qed
\end{proof}

\subsection{Data tree}
Let $(\D,G)$ be a data symmetry and $A$ be an alphabet.
We define an \emph{$m$-ary data tree} (simply \emph{tree}) over $A$ 
as a function $t:\Pos(t) \to A$ 
satisfying the following two conditions:
\begin{itemize}
	\item $\Pos(t) \subseteq \{1,\ldots,m\}^*$ is a non-empty finite set
		  that is prefix-closed, and
	\item every $p \in \Pos(t)$ has a non-negative integer $\arity(p) \leq m$ satisfying
		  \[
			  p\,i \in \Pos(t) \qquad \text{for all $i \in \{1 ,\ldots, \arity(p)\}$},
		  \]
\end{itemize}
where $p \ i$ is the concatenation of $p$ and $i$.
The set of all $m$-ary data trees over $A$ is written as $\Tree_m(A)$.

We define subtree $t|_p$ of $t$ at $p \in \Pos(t)$ as
\begin{itemize}
	\item $\Pos(t|_p) = \{ q \in \{1,\ldots,m\}^* \mid pq \in \Pos(t)\}$, and
	\item $t|_p(q) = t(pq)$ for all $q \in \Pos(t|_p)$.
\end{itemize}
The set of all subtrees of $t$ is written as $\Subtree(t)$.
To denote a tree, we will use term representation, which is recursively defined as follows.
For $a \in A$ and terms $\trm_1,\ldots,\trm_k$ with $0 \leq k \leq m$,
the term $a(\trm_1,\ldots,\trm_k)$ represents the tree $t$ such that
$\arity(\varepsilon)=k$ and
\[
	t(p) = \begin{cases}
				a      & \text{if $p = \varepsilon$}, \\
				t_i(q) & \text{if $p = iq$ \quad for $1 \leq i \leq k$},
		   \end{cases}
\]
where $t_i$ is the tree represented by $\trm_i$ for $1 \leq i \leq k$.

The group action on $\Tree_m(A)$ is defined as 
$t \cdot \pi = (a \cdot \pi)(t_1 \cdot \pi,\ldots,t_k \cdot \pi)$
for all $t \in \Tree_m(A)$ and $\pi \in G$.
$\Tree_m(A)$ is a nominal set.

Let $x \notin A$ be a variable.
A tree $t \in \Tree_m(A \cup \{x\})$ is called a \emph{context} of $A$
if and only if there is exactly one $p \in \Pos(t)$ satisfying
$t(p) = x$ and $\arity(p) = 0$.
The set of all contexts of $A$ is written as $\Context_m(A)$.
For $s \in \Context_m(A)$ and $t \in (\Tree_m(A) \cup \Context_m(A))$,
we define $s[t]$ as a tree with $x$ in $s$ replaced by $t$, i.e.,
\[
	s[t](p) = \begin{cases}
			  		s(p) & \text{if $p \in \Pos(s)$ and $s(p) \neq x$}, \\
					t(q) & \text{if $p = rq,s(r) = x$ and $q \in \Pos(t)$}.
			  \end{cases}
\]
For $S \subseteq \Context_m(A)$ and $T \subseteq (\Tree_m(A) \cup \Context_m(A))$,
we define $S[T] = \{ s[t] \mid \text{$s \in S$ and $t \in T$} \}$.
For all $\pi \in G$, we define $x \cdot \pi = x$.

\begin{example}
	Let $c \in \Context_m$ be a context of $\N$ such that 
	$\Pos(c) = \{ \varepsilon,1\}$, $c(\varepsilon)=2$, $c(1)=x$.
	Let $t \in \Tree_2(\N)$ be a $2$-ary tree over $\N$ such that 
	$\Pos(t) = \{ \varepsilon,1,2\}$, $t(\varepsilon)=3$, $t(1)=1$, $t(2)=5$.
	For $c$ and $t$, $c[t]$ is the tree such that 
	$\Pos(c[t]) = \{\varepsilon,1,11,12\}$, $c[t](\varepsilon)=2$, $c[t](1)=3$, $c[t](11)=1$, $c[t](12)=5$.
	Figure \ref{fig:tree} illustrates $c$, $t$ and $c[t]$.
	Term representations of $c$, $t$ and $c[t]$ are $c=2(x)$, $t=3(1,5)$ and $c[t]=2(3(1,5))$, respectively.
	The term representation of subtree $c[t]|_{12}$ of $c[t]$ at $12 \in \Pos(c[t])$ is $5$.
	\begin{figure}\label{fig:tree}
		\centering
		\begin{tikzpicture}[node distance=8mm]
			\node (a) {$c=$};
			\node[right of=a,node distance=5mm] (b) {$ \begin{forest} [2 [$x$]] \end{forest}$};
			\node[right of=b,node distance=15mm] (c) {$t=$};
			\node[right of=c] (d) {$ \begin{forest} [3 [1] [5]] \end{forest}$};
			\node[right of=d,node distance=15mm] (e) {$c[t]=$};
			\node[right of=e,node distance=10mm] (f) {$ \begin{forest} [2 [3 [1] [5]]] \end{forest}$};
		\end{tikzpicture}
		\caption{Data trees and a context over $\N$.}
	\end{figure}
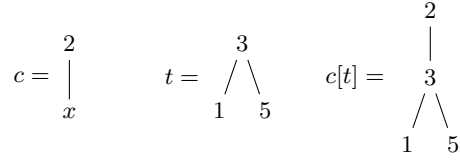
\end{example}

\section{Deterministic bottom-up nominal tree automata}
	Let $(\D,G)$ be a data symmetry and $A$ be an alphabet.
A \emph{deterministic bottom-up tree automaton} ($G$-DBTA)
over $\Tree_m(A)$ is a triple $\A = (Q,F,\delta)$, where
\begin{itemize}
	\item $Q$ is a $G$-set of states,
	\item $F \subseteq Q$ is an equivariant set of accept states, and
	\item $\delta = (\delta_0,\ldots,\delta_m)$ is an
	      {$m+1$}-tuple of equivariant transition functions, where
		  \begin{gather*}
		  		\delta_0:A \to Q, \\
				\delta_k:A \times Q^k \to Q \quad \text{for $1 \leq k \leq m$}.
		  \end{gather*}
\end{itemize}
We extend $\delta$ to the function on $\Tree_m(A)$ by
\[
	\delta(a(t_1,\ldots,t_k)) 
		= \begin{cases}
				\delta_k(a,\delta(t_1),\ldots,\delta(t_k)) & \text{if $k>0$}, \\
				\delta_0(a)                                & \text{if $k=0$}.
		  \end{cases}
\]
A tree $t$ is accepted by $\A$ if and only if $\delta(t) \in F$.
We define $L(\A) = \{ t \in \Tree_m(A) \mid \delta(t) \in F \}$.
We call $L \subseteq \Tree_m(A)$ a \emph{recognizable tree language}
when there exists a $G$-DBTA $\A$ satisfying $L = L(\A)$.
A $G$-DBTA $\A = (Q,F,\delta)$ is \emph{reachable} if 
for all $q \in Q$, there exists some $t \in \Tree_m(A)$ such that $\delta(t) = q$.
If $A$ and $Q$ are orbit finite nominal sets, then $\A$ is called 
a \emph{deterministic bottom-up nominal tree automaton} ($G$-DBNTA).
We call $L \subseteq \Tree_m(A)$ a \emph{recognizable nominal tree language}
when there exists a $G$-DBNTA $\A$ satisfying $L = L(\A)$.

Let $X \subseteq \Tree_m(A)$ be a subset of trees.
For a function $T:X \to \{0,1\}$, a $G$-DBTA $\A$ is \emph{consistent with} $T$
if for all $t \in X$, $t \in L(\A) \Leftrightarrow T(t) = 1$.

\begin{example}
	Let the set $\N$ of natural numbers be an alphabet and $(\N,\Sym(\N))$ be a data symmetry.
	Let $\A = (Q,F,\delta)$ be a $G$-DBTA over $\Tree_2(\N)$, where
	$Q=\N \cup \{\accept,\reject\}$, $F=\{\accept\}$ and
	$\delta=(\delta_0,\delta_1,\delta_2)$ such that
	$\delta_0(d)=d$, $\delta_1(d,q)=\reject$ and
	$\delta_2(d,q_1,q_2)=\accept$ if $q_1=q_2=d$, $\reject$ otherwise.
	For all $\pi \in \Sym(\N)$, we define $\accept \cdot \pi = \accept$
	and $\reject \cdot \pi = \reject$.
	It is easy to see that all components of $\A$ are equivariant.
	The tree language recognized by $\A$ is
	$L(\A) = \{ d(d,d) \mid d \in \N\}$.
\end{example}

\section{Myhill-Nerode theorem}
	Let $(\D,G)$ be a data symmetry that admits least support and $A$ be a nominal alphabet.
For $L \subseteq \Tree_m(A)$,
we define the binary relation $\approx_L$ over $\Tree_m(A)$ as follows: 
$u \approx_L v$ if and only if
\[
	\text{$C[u] \in L$ iff $C[v] \in L$ for all $C \in \Context_m(A)$}.
\]
It is easy to check that $\approx_L$ is an congruence relation on $\Tree_m(A)$,
i.e., $\approx_L$ is an equivalent relation that satisfies
$a(u_1,\ldots,u_k) \approx_L a(v_1,\ldots,v_k)$ for all
$a \in A$ and $u_1,\ldots,u_k,v_1,\ldots,v_k \in \Tree_m(A)$ 
with $u_i \approx_L v_i$ for $0 \leq i \leq k$.
We write the equivalent class of $t \in \Tree_m(A)$ as $[t]$.

\begin{lemma}\label{lem:approx}
	If $L \subseteq \Tree_m(A)$ is equivariant, then $\approx_L$ is also equivariant.
\end{lemma}

\begin{proof}
	We show that $t \cdot \pi \approx_L t' \cdot \pi$
	for all $\pi \in G$ and $t,t' \in \Tree_m(A)$ such that $t \approx_L t'$.
	By the definition of $\approx_L$, $t \cdot \pi \approx_L t' \cdot \pi$ is equivalent to
	$C[t \cdot \pi] \in L$ iff $C[t' \cdot \pi] \in L$ for all $C \in \Context_m(A)$.
	By the equivariance of $L$, this is equivalent to
	$C[t \cdot \pi] \cdot \pi^{-1} \in L$ 
	iff $C[t' \cdot \pi] \cdot \pi^{-1} \in L$.
	By the definition of the group action on $\Tree_m(A)$, this is equivalent to
	$(C \cdot \pi^{-1})[t] \in L$ iff $(C \cdot \pi^{-1})[t'] \in L$.
	We can prove this by $t \approx_L t'$.
	\qed
\end{proof}

\begin{lemma}[{\cite[Lemma 3.5]{BKL14}}]\label{lem:quotient}
	Let $X$ be a $G$-set and $R \subseteq X \times X$ be an equivalence relation
	that is equivariant.
	Then the quotient $X/R$ is a $G$-set, under the action
	$[x]_R \cdot \pi = [x \cdot \pi]_R$ of $G$, and the abstraction mapping
	$x \mapsto [x]_R \ : \ X \to X/R$ is an equivariant function.
\end{lemma}

For $L \subseteq \Tree_m(A)$, 
we define the \emph{syntactic tree automaton} $\A_L = (Q_L,F_L,\delta_L)$ as
\begin{itemize}
	\item $Q_L = \Tree_m(A)/{\approx_L}$,
	\item $F_L = \{ [t] \mid t \in L \}$ and
	\item $\delta_L = (\delta_{L,0},\ldots,\delta_{L,m})$ where 
		  $\delta_{L,0}:A \to Q_L$ and $\delta_{L,k}:A \times Q_L^k \to Q_L$ 
		  for $1 \leq k \leq m$ are defined as
		\begin{align*}
			\delta_{L,0}(a) &= [a], \\
			\delta_{L,k}(a,[u_1],\ldots,[u_k]) &= [a(u_1,\dots,u_k)].
		\end{align*}
\end{itemize}
Because $\approx_L$ is a congruence relation, $\delta_L$ is well-defined.

\begin{lemma}\label{lem:syntactic}
	If $L \subseteq \Tree_m(A)$ is equivariant, then
	the syntactic tree automaton $\A_L = (Q_L,F_L,\delta_L)$ is a reachable $G$-DBTA.
\end{lemma}

\begin{proof}
	Because $L$ is equivariant, $\approx_L$ is also equivariant by Lemma \ref{lem:approx}.
	Thus, by Lemma \ref{lem:quotient}, $Q_L = \Tree_m(A)/{\approx_L}$ is a $G$-set.
	By the equivariance of $L$,
	\[
		[t] \in F_L \Leftrightarrow t \in L \Leftrightarrow t \cdot \pi \in L
			\Leftrightarrow [t \cdot \pi] \in F_L \Leftrightarrow [t] \cdot \pi \in F_L.
	\]
	Thus, $F_L$ is equivariant.
	By
		$\delta_L(a,[u_1],\ldots,[u_k]) \cdot \pi 
			= [a(u_1,\ldots,u_k)] \cdot \pi 
			= [(a \cdot \pi)(u_1 \cdot \pi,\ldots,u_k \cdot \pi)] 
			= \delta_L(a \cdot \pi,[u_1 \cdot \pi],\ldots,[u_k \cdot \pi]) 
			= \delta_L(a \cdot \pi,[u_1] \cdot \pi,\ldots,[u_k] \cdot \pi)$,
	$\delta_L$ is equivariant.
	Thus, $\A_L$ is a $G$-DBTA.
	$\A_L$ is apparently reachable.
	\qed
\end{proof}

Let $\A = (Q,F,\delta)$ and $\A' = (Q',F',\delta')$ be $G$-DBTAs.
An equivariant function $\varphi:P \to Q'$ for some $P \subseteq Q$ 
satisfying the following two conditions 
is called a \emph{partial homomorphism} from $\A$ to $\A'$:
\begin{itemize}
	\item $q \in F$ iff $\varphi(q) \in F'$ for all $q \in P$, and
	\item $\varphi(\delta(a,q_1,\ldots,q_k)) = \delta'(a,\varphi(q_1),\ldots,\varphi(q_k))$ 
		for all $q_1,\ldots,q_k \in P \ (0 \leq k \leq m)$ and $a \in A$.
\end{itemize}
When there exists a surjective partial homomorphism from a subset of $Q$ to $Q'$, 
we write $\A' \sqsubseteq \A$. 
If $P = Q$, then $\varphi$ is called a \emph{homomorphism} from $\A$ to $\A'$.
It is easy to see that $L(\A) = L(\A')$ 
if there is a homomorphism $\varphi$ from $\A$ to $\A'$.
When $\varphi$ is surjective, $\A'$ is called an image of $\A$.
If $\A'$ is an image of $\A$ and the state set of $\A$ is orbit finite,
the state set of $\A'$ is also orbit finite.

\begin{lemma}\label{lem:image}
	Let $L$ be a recognizable tree language.
	The syntactic automaton $A_L$ is an image of any reachable $G$-DBTA recognizing $L$.
\end{lemma}

\begin{proof}
	Let $\A = (Q,F,\delta)$ be a reachable $G$-DBTA recognizing $L$.
	We define $\varphi:Q \to \Tree_m(A)/{\approx_L}$ as $\varphi(\delta(t)) = [t]$.
	The definition of $\varphi$ is well-defined because
	$\A$ is reachable and $\delta(u) = \delta(v)$ implies $[u] = [v]$.
	It is easy to check that $f$ is surjective.
	By the equivariance of $\delta$,
	\[
		\varphi(\delta(t) \cdot \pi) = \varphi(\delta(t \cdot \pi)) 
			= [t \cdot \pi] = [t] \cdot \pi = \varphi(\delta(t)) \cdot \pi.
	\]
	Thus, $\varphi$ is equivariant.
	We have
	\begin{align*}
		&\varphi(\delta(a(u_1,\ldots,u_k))) 
				= [a(u_1,\ldots,u_k)] \\
				&= \delta_L(a,[u_1],\ldots,[u_k]) 
				= \delta(a,\varphi(\delta(u_1)),\ldots,\varphi(\delta(u_k))).
	\end{align*}
	We also have $\delta(t) \in F \Leftrightarrow t \in L \Leftrightarrow [t] \in F_L$.
	Therefore, $\varphi$ is a homomorphism, and hence $A_L$ is an image of $\A$.
	\qed
\end{proof}

Let $\A = (Q,F,\delta)$ be a reachable $G$-DBTA over $\Tree_m(A)$.
The equivariant function $t \mapsto \delta(t)$ from $\Tree_m(A)$ to $Q$
is surjective because $\A$ is reachable.
If $C \subseteq \D$ supports $t$, then $C$ also supports $\delta(t)$ because
$t \cdot \pi = t$ implies $\delta(t) \cdot \pi = \delta(t \cdot \pi) = \delta(t)$
for all $\pi \in G$.
Thus, $Q$ is nominal because $\Tree_m(A)$ is nominal.
\begin{theorem}
	Let $L \subseteq \Tree_m(A)$ be an equivariant set.
	The following two conditions are equivalent:
	\begin{enumerate}
		\item[(1)] $\Tree_m(A)/{\approx_L}$ is orbit finite.
		\item[(2)] $L$ is recognized by a $G$-DBNTA.
	\end{enumerate}
\end{theorem}

\begin{proof}
	$(1) \Rightarrow (2)$ can be easily proved by Lemma \ref{lem:syntactic}.
	Without loss of generality, assume that a given $G$-DBNTA is reachable.
	$(2) \Rightarrow (1)$ can be proved by Lemma~\ref{lem:image}.
\qed
\end{proof}

\section{Observation table}
	In this and the next sections, we extend the $L^{\ast}$-style algorithm in \cite{MSSKS17} to DBNTA.
For the extension from words to trees, we extend some notions given in \cite{Sa90}, where
another $L^{\ast}$-style algorithm is proposed to learn the set of derivation trees of an unknown context-free grammar without data values.

From now on, we assume a data symmetry $(\D,G)$ that admits least support
and a nominal alphabet $A$.
Let $B \subseteq \Tree_m(A)$ and $C \subseteq \Context_m(A)$.
We say that $B$ is \emph{subtree-closed} if and only if 
for every $b \in B$, $\Subtree(b) \subseteq B$ holds.
We also say that $C$ is \emph{$x$-prefix-closed} on $B$ if and only if
every $c \in C \setminus \{x\}$ has some $c' \in C$ satisfying
$c = c'[a(b_1,\ldots,b_{i-1},x,b_i,\ldots,b_{k-1})]$ where
$b_1,\ldots,b_{k-1} \in B$ and $a \in A$.

\begin{definition}
	Let $U$ be an unknown recognizable nominal tree language.
	An observation table is a triple $(\S,\E,\T)$, where
	\begin{itemize}
		\item $\S \subseteq \Tree_m(A)$ 
		      is an equivariant orbit finite set that is subtree-closed
			  and satisfies $A \subseteq \S$.
		\item $\Next(\S) = \{a(t_1,\ldots,t_k) \notin \S \mid  
				a \in A,t_1,\ldots,t_k \in \S, 1 \leq k \leq m\}$,
		\item $\E \subseteq \Context_m(A)$ is an equivariant orbit finite set
			  that is $x$-prefix-closed on $\S$, and
		\item $\T:\E[\S \cup \Next(\S)] \to \{0,1\}$ is an equivariant function, where
			  $\T(e[s]) = 1$ iff $e[s] \in U$
			  for all $e \in \E$ and $s \in \S \cup \Next(\S)$.
	\qed
	\end{itemize}
\end{definition}
We define the function $\row_{(\S,\E,\T)}:\S \cup \Next(\S) \to 2^{\E}$ as
$\row_{(\S,\E,\T)}(s) = \{ e \in \E \mid \T(e[s]) = 1 \}$.
We abbreviate $\row_{(\S,\E,\T)}$ as $\row$ if $(\S,\E,\T)$ is clear from the context.
It is easy to see that $\row$ is equivariant.
For $X \subseteq \S \cup \Next(\S)$, we define $\row(X) = \{ \row(s) \mid s \in X \}$.

An observation table can be expressed by the table 
with rows labeled with the elements of $\S \cup \Next(\S)$ and
columns labeled with the elements of $\E$ as shown in Fig. \ref{fig:ob-table}.
\begin{figure}
	\[
		\begin{array}{|cc|ccc|}
			\hline
			 & & \rule{3em}{0em}  & \E & \rule{3em}{0em} \\
			  & &  & e & \\
			\hline
			\rule{0em}{2em} &&& \vdots & \\
			\S & s & \cdots & \T(e[s]) & \\
			\rule{0em}{2em} &&&& \\
			\hline
			\rule{0em}{2em} &&&& \\
			\Next(\S) & & & &\\
			\rule{0em}{2em} &&&& \\
			\hline
		\end{array}
	\]
	\caption{Observation table $(\S,\E,\T)$}
	\label{fig:ob-table}
\end{figure}

An observation table $(\S,\E,\T)$ is \emph{closed} if and only if
for all $t \in \Next(\S)$, there exists some $s \in \S$ satisfying
$\row(t) = \row(s)$.
An observation table $(\S,\E,\T)$ is \emph{consistent} if and only if
for every $s_1,s_2 \in \S$, $\row(s_1) = \row(s_2)$ implies
\[
	\row(a(u_1,\ldots,u_{i-1},s_1,u_i,\ldots,u_{k-1})) = 
	\row(a(u_1,\ldots,u_{i-1},s_2,u_i,\ldots,u_{k-1}))
\]
for all $a \in A,u_1,\ldots,u_{k-1} \in \S$ and $1 \leq i \leq k$.

Let $(\S,\E,\T)$ be a closed and consistent observation table.
We define the $G$-DBNTA $\A(\S,\E,\T) = (\Q,\F,\del)$ derived from $(\S,\E,\T)$ as follows:
\begin{itemize}
	\item $\Q = \row(\S) = \{ \row(s) \mid s \in \S \}$,
	\item $\F = \{ \row(s) \mid s \in \S,\T(s) = 1 \}$,
	\item $\del_k(a,\row(s_1),\ldots,\row(s_k)) = \row(a(s_1,\ldots,s_k))$
			for $s_1,\ldots,s_k \in \S$.
\end{itemize}
It is easy to see that $\A(\S,\E,\T)$ is well-defined:
Let $s_1,s_2 \in \S$ be trees satisfying $\row(s_1) = \row(s_2)$.
Because $\E$ is $x$-prefix-closed, $x \in \E$ holds.
Thus, $\T(s_1) = \T(x[s_1])$ and $\T(s_2) = \T(x[s_2])$ are defined,
and $\T(s_1) = \T(s_2)$, and so $\F$ is well-defined.
Because $(\S,\E,\T)$ is consistent,
\[
	\row(a(u_1,\ldots,u_{i-1},s_1,u_{i},\ldots,u_{k-1})) 
	= \row(a(u_1,\ldots,u_{i-1},s_2,u_{i},\ldots,u_{k-1}))
\]
for all $a \in A,u_1,\ldots,u_{k-1}\in \S$ and $1 \le i \le k$.
Moreover, because $(\S,\E,\T)$ is closed,
there is $s \in \S$ satisfying
$\row(s) = \row(a(u_1,\ldots,u_{i-1},s_1,u_{i},\ldots,u_{k-1}))$.
Therefore, $\del$ is well-defined.
Because $\S$ is an orbit finite nominal set and $\row$ is an equivarinat function, 
$\Q(=\row(\S))$ is also an orbit finite nominal set.

\begin{lemma}\label{lem:consistency}
	Let $(\S,\E,\T)$ be a closed and consistent observation table.
	Then, $\A(\S,\E,\T) = (\Q,\F,\del)$ is consistent with $\T$.
\end{lemma}
The proof is similar to the proof of Lemma 4.2 in \cite{Sa90}.

\begin{theorem}\label{thm:larger}
	For a (not necessarily closed and consistent) observation table $(\S,\E,\T)$ and
	a $G$-DBNTA $\A = (Q,F,\delta)$ consistent with $\T$,
	$\row(\S) \preceq Q$ holds.
\end{theorem}

\begin{proof}
	We show that the function $\delta(s) \mapsto \row(s)$ is an equivariant surjection
	from $\{ \delta(s) \mid s \in \S \} \subseteq Q$ to $\row(\S)$.
	This function is well-defined since
	\begin{align*}
			\delta(s_1) = \delta(s_2) 
			& \Rightarrow \forall e \in \E.\delta(e[s_1]) = \delta(e[s_2]) \\
			& \Rightarrow \forall e \in \E.e[s_1]\in L(\A) \quad 
				\mathit{iff} \quad e[s_2] \in L(\A) \\
			& \Leftrightarrow \forall e \in \E.\T(e[s_1]) = \T(e[s_2]) \\
			& \Leftrightarrow \row(s_1) = \row(s_2).
	\end{align*}
	This function is also equivariant because
	$\delta(s) \cdot \pi = \delta(s \cdot \pi)
		\mapsto \row(s \cdot \pi) = \row(s) \cdot \pi$.
	Surjectivity is clear.
	Therefore, because there is an equivariant function from a subset of $Q$ to $\row(\S)$,
	$\row(\S) \preceq Q$ holds.
	\qed
\end{proof}

If an observation table is closed and consistent, Theorem \ref{thm:larger} can be lifted
from a relation on states ($\preceq$) to a relation on automata ($\sqsubseteq$)
as stated in the next lemma.

\begin{lemma}\label{lem:smallestness}
	Let $(\S,\E,\T)$ be a closed and consistent observation table.
	For every $G$-DBNTA $\A$ that is consistent with $\T$,
	$\A(\S,\E,\T) \sqsubseteq \A$ holds.
\end{lemma}

\begin{proof}
	Let $\A(\S,\E,\T) = (\Q,\F,\del)$ and $A = (Q,F,\delta)$.
	By the proof of Theorem \ref{thm:larger},
	the function $\varphi:\delta(s) \mapsto \row(s)$ 
	from $\{ \delta(s) \mid s \in \S \} \subseteq Q$ to $\Q(=\row(\S))$
	is equivariant and surjective.
	By
	$\delta(s) \in F \Leftrightarrow \T(s) = 1
		\Leftrightarrow \row(s) \in \F 
		\Leftrightarrow \varphi(\delta(s)) \in \F$
	and
	$\varphi(\delta(a,\delta(s_1),\ldots,\delta(s_k))) 
			= \varphi(\delta(a(s_1,\ldots,s_k)))
		  	= \row(a(s_1,\ldots,s_k)) 
		  	= \del(a,\row(s_1),\ldots,\row(s_k))
		  	= \del(a,\varphi(\delta(s_1)),\ldots,\varphi(\delta(s_k)))$,
	$\varphi$ is a partial homomorphism.
	\qed
\end{proof}

By Lemmas \ref{lem:consistency} and \ref{lem:smallestness},
we have the following theorem.

\begin{theorem}
	Let $(\S,\E,\T)$ be a closed and consistent observation table.
	$\A(\S,\E,\T)$ is consistent with $\T$, and
	for every $G$-DBNTA $\A$ that is consistent with $\T$,
	$\A(\S,\E,\T) \sqsubseteq \A$ holds.
\end{theorem}

\section{Learning algorithm}
	We show the proposed learning algorithm (Algorithm~\ref{alg:learn}) in the following page.
We will give a part of a run of Algorithm~\ref{alg:learn} on an example in Section~\ref{sec:example}.
In the Algorithm~\ref{alg:learn}, we assume that the teacher answering queries is given as an oracle.
In an application to the compositional verification, for example, the teacher is implemented as a model checker (see \cite{CGP03}).

Because $\S$ and $\E$ of an observation table $(\S,\E,\T)$ can be infinite sets,
we have to show that $(\S,\E,\T)$ can be expressed by finite means
and each step of Algorithm \ref{alg:learn} runs in finite steps.
We first show that $(\S,\E,\T)$ has a finite description.
Because $\S$ and $\E$ are orbit finite nominal sets, by Proposition \ref{prop:CSec},
we can express $\S$ and $\E$ by support representations.
By Proposition \ref{prop:func}, we can express $\T$ by finite means because
$\T$ consists of a finite number of equivariant functions 
whose domains and ranges are both single orbit nominal sets.
In Algorithm~\ref{alg:learn}, each orbit $O$ is represented by any one element $s \in O$.
Let us call $s$ a representative of $O$.

Next, we show that each step of Algorithm \ref{alg:learn} runs in finite steps.
To check the closedness in line 10 of Algorithm~\ref{alg:learn},
it suffices to check 
whether for each orbit $O$ of $\Next(\S)$ and a representative $s'$ of $O$,
there is $s \in \S$ such that $\row(s) = \row(s')$.
Finding $s \in \S$ satisfying $\row(s) = \row(s')$ is equivalent to
finding $\pi \in G$ such that $\row(s') = \row(t \cdot \pi)\,(=\row(t) \cdot \pi)$ 
for some representative $t \in \S$.
Let $C,D \subseteq \D$ be the least supports of $\row(s')$ and $\row(t)$, respectively.
The least support of $\row(t) \cdot \pi$ is $D \cdot \pi$.
Thus, if $\row(s') = \row(t \cdot \pi)$, then $C = D \cdot \pi$ must hold.
Moreover, because $D$ is the (least) support of $\row(t)$,
if $\pi_1|_D = \pi_2|_D$ then $\row(t) \cdot \pi_1 = \row(t) \cdot \pi_2$. 
Thus, we only have to check a finite number of $\pi$ satisfying $C = D \cdot \pi$.
To check the consistency in line 5 of Algorithm~\ref{alg:learn},
it suffices to check the emptiness of
\begin{align*}
	&\{(s_1,s_2,a,e) \in \S \times \S \times A \times \E \mid 
		\text{$\row(s_1) = \row(s_2)$ and for $\exists u_1,\ldots,u_{k-1} \in \S$,} \\
	&\T(e[a(u_1,\ldots,u_{i-1},s_1,u_i,\ldots,u_{k-1})]) \neq
		T(e[a(u_1,\ldots,u_{i-1},s_2,u_i,\ldots,u_{k-1})])\}.
\end{align*}
$\S \times \S \times A \times \E$ is an orbit finite nominal set,
and the above set is a union of some orbits of $\S \times \S \times A \times \E$.
Thus, we can check the emptiness of the set, and if not, we can obtain representatives of the set.
\begin{algorithm}[t]
  \caption{Angluin-style algorithm for \textit{G-DBNTA}}
  \label{alg:learn}
\begin{algorithmic}[1]
\State $\S:=A, 
		\E:=\{x\}$;
\State Construct the initial observation table $(\S,\E,\T)$ using membership queries;
\Repeat 
\While {$(\S,\E,\T)$ is not closed or not consistent}
	\If{$(\S,\E,\T)$ is not consistent}
	\State Find ${s_1,s_2,u_1,\ldots,u_{k-1} \in \S},{e \in \E},{a \in A},{i \in \mathbb{N}}$ such that
			\begin{align*}
				&\text{$\row(s_1) = \row(s_2)$ and} \\
				&{T(e[a(u_1,\ldots,u_{i-1},s_1,u_i,\ldots,u_k)]) 
					\neq T(e[a(u_1,\ldots,u_{i-1},s_2,u_i,\ldots,u_k)])};
			\end{align*}
	\State Add $\Orbit(e[a(u_1,\ldots,u_{i-1},x,u_i,\ldots,u_k)])$ to $\E$;
	\State Extend $\T$ to $\E[(\S \cup \Next(\S))]$ using membership queries;
	\EndIf

	\If{$(\S,\E,\T)$ is not closed}
	\State Find $s' \in \Next(\S)$ such that $\row(s') \neq \row(s) \text{ for all } s \in \S$;
	\State Add $\Orbit(s')$ to $\S$;
	\State Extend $\T$ to $\E[\S \cup \Next(\S)]$ using membership queries;
	\EndIf
\EndWhile
\State Let $\A = \A(\S,\E,\T)$;
\State Construct the conjecture $\A$;
\If{the Teacher replies \textit{no} with a counter-example $t$}
	\State Add $\Orbit(\Subtree(t))$ to $\S$;
	\State Extend $\T$ to $\E[\S \cup \Next(\S)]$ using membership queries;
\EndIf
\Until{the Teacher replies \textit{yes} to the conjecture $\A$;}
\State \Return $\A$;
\end{algorithmic}
\end{algorithm}

\paragraph*{Correctness}
Because Algorithm \ref{alg:learn} uses an equivalence query,
if it terminates, then it outputs the correct $G$-DBNTA.

To prove the termination of Algorithm~\ref{alg:learn},
we show the following two lemmas that guarantee that 
$\row(\S)$ strictly increases with respect to $\prec$
each time an observation table is extended.

\begin{lemma}\label{lem:addS}
	If $(\S,\E,\T)$ and $(\S',\E,\T')$ are observation tables such that
	$\S \subsetneq \S'$ and 
	$\T'(e[s]) = \T(e[s])$ for all $e \in \E$ and $s \in \S$,
	then $\row_{(\S,\E,\T)}(\S) \prec \row_{(\S',\E,\T')}(\S')$.
\end{lemma}

\begin{proof}
	The lemma obviously holds because if $\S \subsetneq \S'$,
	the number of orbits of $\S'$ is larger than that of $\S$.
	\qed
\end{proof}

\begin{lemma}\label{lem:addE}
	If $(\S,\E,\T)$ and $(\S,\E',\T')$ are observation tables such that
	$\E \subsetneq \E'$ and 
	$\T'(e[s]) = \T(e[s])$ for all $e \in \E$ and $s \in \S$, then
	$\row_{(\S,\E,\T)}(\S) \prec \row_{(\S,\E',\T')}(\S)$.
\end{lemma}
A proof of this lemma is given in the Appendix.

\paragraph*{Termination and minimality}
Let $U$ be an unknown recognizable nominal tree language and
$\A_U = (Q_U,F_U,\delta_U)$ be the syntactic tree automaton constructed from $U$.
$\A_U$ is the minimum $G$-DBNTA recognizing $U$ in the sence of Lemma~\ref{lem:image}.
Let $(\S_0,\E_0,\T_0),(\S_1,\E_1,\T_1),(\S_2,\E_2,\T_2),\ldots$ be observation tables
constructed by Algorithm \ref{alg:learn} where
$(\S_i,\E_i,\T_i)$ extends to $(\S_{i+1},\E_{i+1},\T_{i+1})$ for $i \geq 0$.
Note that $\A_U$ is consistent with every $\T_i$ for $i \geq 0$.
By Lemmas \ref{lem:addS} and \ref{lem:addE},
$\row(\S_0) \prec \row(\S_1) \prec \row(\S_2) \prec \cdots$.
By Theorem \ref{thm:larger}, $\row(\S_i) \preceq Q_U$ for $i \geq 0$.
By Lemma \ref{lem:finiteness-f}, there is a non-negative integer $n$ such that
$\row(\S_0) \prec \row(\S_1) \prec \cdots \prec \row(\S_n) = Q_U$.
Thus, Algorithm \ref{alg:learn} terminates in finite steps.
By Lemma \ref{lem:smallestness}, 
$\A(\S_i,\E_i,\T_i) \sqsubseteq \A_U$ holds for every $i \geq 0$
such that $(\S_i,\E_i,\T_i)$ is closed and consistent,
and hence, Algorithm \ref{alg:learn} 
outputs the minimum $G$-DBNTA recognizing $U$ when it terminates.

\paragraph*{Running time analysis}
When Algorithm~\ref{alg:learn} extends an observation table $(\S,\E,\T)$,
the number of orbits of $\row(\S)$ increases or some orbits of $\row(\S)$ extend.
Extending an orbit $\CSec$ of $\row(\S)$ to $\DTec$ implies $|C| \leq |D|$.
If $|C| = |D|$, then $T \leq uSu^{-1}$ must hold for some injection $u:D \to C$.
By the standard theorem of finite groups,
$|uSu^{-1}| (=|S|)$ can be divided by $|T|$.
Therefore, we have the following theorem:
\begin{theorem}
Let $U$ be an unknown recognizable nominal tree language,
$Q = {\lec C_1,S_1 \rec} \cup \cdots \cup {\lec C_n,S_n \rec}$
be the set of states of the minimum $G$-DBNTA recognizing $U$,
$n$ be the number of orbits of $Q$ and
$m = \max\{|C_1|,\ldots,|C_n|\}$
be the largest cardinality of least supports of orbits of $Q$.
Let $p_1,\ldots,p_k$ be prime numbers and $j_1,\ldots,j_k$ be positive integers such that
$m! = p_1^{j_1} \cdot p_2^{j_2} \cdot \cdots \cdot p_k^{j_k}$.
Observation tables are extended at most $O(nm(j_1 + \cdots + j_k))$ times.
\end{theorem}

\section{Example}\label{sec:example}
	Let $(\N,\Sym(\N))$ be the equality symmetry and $A = \N$ be an alphabet.
Note that $(\N,\Sym(\N))$ admits least support and $A$ is a single orbit nominal set.
Let $U = \Orbit(1) \cup \Orbit(1(1)) \cup \Orbit(1(1(1))) \subseteq \Tree_2(A)$.
We now show a part of a run of Algorithm~\ref{alg:learn} for $U$.

First, the elements of the initial observation table $(\S_0,\E_0,\T_0)$ 
shown in Table~\ref{tab:first} are
$\S_0 = \Orbit(1) \, (=\N)$,\;
$\E_0 = \{x\}$,\;
$\Next(\S_0) = \Orbit(1(1)) \cup \Orbit(2(1)) \cup \Orbit(1(1,1)) \cup \Orbit(1(2,1))
		\cup \Orbit(1(1,2)) \cup \Orbit(2(1,1))$,
	$\T_0(a) = \T_0(a(a)) = 1$ and
	$\T_0(a(b)) = \T_0(a(a,a)) = \T_0(a(b,a)) = \T_0(a(a,b)) =
		\T_0(a(b,b)) = 0$
for all $a,b \in A$ such that $a \neq b$.
This observation table $(\S_0,\E_0,\T_0)$ is consistent but not closed because 
there is no $s \in \S_0$ such that $\row(s) = \row(2(1))$.
Thus, Algorithm~\ref{alg:learn} adds $\Orbit(2(1))$ to $\S_0$ and
extends $\T_0$ using membership queries.
We have the observation table $(\S_1,\E_1,\T_1)$ shown in Table~\ref{tab:second} where
\begin{gather*}
	{\S_1 = \Orbit(1) \cup \Orbit(2(1))},\quad{\E_1 = \{x\}}, \\
	\Next(\S_1) = \{a(t) \notin \S_1 \mid a \in A,t \in \S_1\} 
		\cup \{a(t_1,t_2) \notin \S_1 \mid a \in A,t_1,t_2 \in \S_1\}.
\end{gather*}

$(\S_1,\E_1,\T_1)$ is closed and consistent, and 
Algorithm~\ref{alg:learn} asks an equivalence query with $G$-DBNTA $\A(\S_1,\E_1,\T_1)$.
$\A(\S_1,\E_1,\T_1)$ does not recognize $U$ because $1(1(1)) \notin L(\A(\S_1,\E_1,\T_1)$.
Thus, Algorithm~\ref{alg:learn} adds $\Orbit(1(1(1)))$ to $\S_1$ 
if $1(1(1))$ is returned as a counterexample and extends $\T_1$
using membership queries.
We have the observation table $(\S_2,\E_2,\T_2)$ shown in Table~\ref{tab:third}.
$(\S_2,\E_2,\T_2)$ is closed but not consistent because
despite $\row(1) = \row(1(1(1)))$, $\row(1(1)) \neq \row(1(1(1(1))))$.
Thus, Algorithm~\ref{alg:learn} adds $\Orbit(1(x))$ to $\E_2$ and extends $\T_2$
using membership queries,
resulting in the observation table $(\S_3,\E_3,\T_3)$ shown in Table~\ref{tab:fourth}.
Continuing these extensions, Algorithm~\ref{alg:learn} finally obtains an observation table 
$(\S_n,\E_n,\T_n)$ such that $\A(\S_n, \E_n,\T_n)$ recognizes $U$.

\begin{table}
	\begin{minipage}[t]{2.8cm}
	\centering
	\caption{
	}
	\label{tab:first}
	$
	\begin{array}{c|c}
		&x \\
		\hline
		a & 1 \\
		\hline
		a(a) & 1 \\
		a(b) & 0 \\
		a(a,a) & 0 \\
		a(a,b) & 0 \\
		a(b,a) & 0 \\
		a(b,b) & 0 
	\end{array}
	$
	\end{minipage}
	\begin{minipage}[t]{2.8cm}
	\centering
	\caption{
	}
	\label{tab:second}
	$
	\begin{array}{c|c}
		&x \\
		\hline
		a & 1 \\
		a(b) & 0 \\
		\hline
		a(a) & 1 \\
		a(a,a) & 0 \\
		a(a,b) & 0 \\
		a(b,a) & 0 \\
		a(b,b) & 0 \\
		\mathit{others} & 0
	\end{array}
	$
	\end{minipage}
	\begin{minipage}[t]{2.8cm}
	\centering
	\caption{
	}
	\label{tab:third}
	$
	\begin{array}{c|c}
		&x \\
		\hline
		a & 1 \\
		a(b) & 0 \\
		a(a(a)) & 1 \\
		\hline
		a(a) & 1 \\
		a(a,a) & 0 \\
		a(a,b) & 0 \\
		a(b,a) & 0 \\
		a(b,b) & 0 \\
		\mathit{others} & 0
	\end{array}
	$
	\end{minipage}
	\begin{minipage}[t]{2.8cm}
	\centering
	\caption{
	}
	\label{tab:fourth}
	$
	\begin{array}{c|cc}
		&x & c(x)\\
		\hline
		a & 1 & {\begin{array}{ll}
					1 & (a = c) \\ 
					0 & (a \neq c)
				\end{array}}\\
		a(b) & 0 & 0 \\
		a(a(a)) & 1 & 0\\
		\hline
		a(a) & 1 & {\begin{array}{ll}
					1 & (a = c) \\ 
					0 & (a \neq c)
				\end{array}}\\
		a(a,a) & 0 & 0\\
		a(a,b) & 0 & 0\\
		a(b,a) & 0 & 0\\
		a(b,b) & 0 & 0\\
		\mathit{others} & 0 & 0
	\end{array}
	$
	\end{minipage}
	\\
	for all $a,b,c \in A$ satisfying $a \neq b$.
\end{table}

\section{Conclusion}
	In this paper, we defined deterministic bottom-up nominal tree automata (DBNTA), 
which operate on trees whose nodes are labelled with elements of an orbit finite nominal set. 
We then proved a Myhill-Nerode theorem for the class of languages recognized by DBNTA 
and proposed an active learning algorithm for DBNTA based on the theorem. 
The algorithm can deal with any data symmetry that admits least support, 
not restricted to the equality symmetry and/or the total order symmetry. 

Implementation and possible applications of the proposed learning algorithm are left as future work.
For implementation, a concrete data structure for support representations of orbit finite sets in
an observation table should be determined.
Moreover, we are considering an application of the proposed algorithm to a compositional verification of a program
that manipulates XML documents.

\appendix
\newpage
\section*{Appendix}
\smallskip\par\noindent
{\bf Proposition \ref{prop:func}.}
{\it
	Let $X = \CSec$ and $Y = \DTec$ be single orbit nominal sets.
	For every equivariant function $f:X \to Y$,
	there is an injection $u$ from $D$ to $C$ satisfying $uS \subseteq Tu$ and
	$f([\pi|_C]_S) = [u]_T \cdot \pi$, for all $\pi \in G$,
	where $uS = \{ us \mid s \in S \}$ and $Tu = \{ tu \mid t \in T \}$.
	Conversely, for every injection $u:D \to C$ satisfying $uS \subseteq Tu$,
	$f([\pi|_C]_S) = [u]_T \cdot \pi$ is an equivariant function from $X$ to $Y$.
}

\begin{proof}
	Let $f:X \to Y$ be an equivariant function.
	Let $u$ be an arbitrary element of the equivalent class $f([e|_C]_S)$; i.e.,
	$f([e|_C]_S) = [u]_T$.
	By \cite[Proposition 9.16]{BKL14}, 
	we can assume the type of $u$ is $u:D \to C$ without loss of generality.
	Then, $f([\pi|_C]_S) = f([e|_C]_S \cdot \pi) = f([e|_C]_S) \cdot \pi = [u]_T \cdot \pi$
	for all $\pi \in G$
	because $f$ is equivariant.
	For any $x \in X$ and $\pi \in G$,
	if $x \cdot \pi = x$, then $f(x) \cdot \pi = f(x \cdot \pi) = f(x)$
	because $f$ is equivariant.
	Therefore, 
	\begin{gather*}
		\forall \pi \in G.\ [e|_C]_S \cdot \pi = [e|_C]_S
			\Rightarrow [u]_T \cdot \pi = [u]_T \\
		\forall \pi \in G.\ [\pi|_C]_S = [e|_C]_S
			\Rightarrow [u\pi]_T = [u]_T \\
		\forall \pi \in G.\ \pi|_C \equiv_S e|_C
			\Rightarrow u\pi \equiv_T u \\
		\forall \pi \in G.\ \pi|_C (e|_C)^{-1} \in S
			\Rightarrow u\pi u^{-1} \in T \\
		\forall \pi \in G.\ \pi|_C \in S \Rightarrow u\pi u^{-1} \in T \\
		\forall \tau \in S.\ u\tau u^{-1} \in T \\
		\forall \tau \in S.\ \tau \in u^{-1} T u \\
		S \leq u^{-1} T u \\
		uS \subseteq Tu.
	\end{gather*}

	Let $u:D \to C$ be an injective function satisfying $uS \subseteq Tu$.
	We prove that $f([\pi|_C]_S) = [u]_T \cdot \pi$ is an equivariant function
	from $X$ to $Y$.
	We first show that $f$ is well-defined.
	When $[\pi|_C]_S = [\sigma|_C]_S$,
	\begin{gather*}
		\pi|_C \equiv_S \sigma|_C \\
		(\pi|_C)(\sigma|_C)^{-1} \in S \\
		u(\pi|_C)(\sigma|_C)^{-1} \in uS \subseteq Tu \\
		u(\pi|_C)(\sigma|_C)^{-1}u^{-1} \in T \\
		(u\pi)|_D (u\sigma)|_D^{-1} \in T \\
		(u\pi)|_D \equiv_T (u\sigma)|_D \\
		[(u\pi)|_D]_T = [(u\sigma)|_D]_T \\
		[u]_T \cdot \pi = [u]_T \cdot \sigma
	\end{gather*}
	Moreover, $f$ is equivariant because
	$f([\pi|_C]_S \cdot \rho) = f([(\pi\rho)|_C]_S) = [u]_T \cdot (\pi\rho) 
	= ([u]_T \cdot \pi) \cdot \rho = f([\pi|_C]_S) \cdot \rho$.
	\qed
\end{proof}

\smallskip\par\noindent
{\bf Lemma \ref{lem:addE}.}
{\it
	If $(\S,\E,\T)$ and $(\S,\E',\T')$ are observation tables such that
	$\E \subsetneq \E'$ and 
	$\T'(e[s]) = \T(e[s])$ for all $e \in \E$ and $s \in \S$, then
	$\row_{(\S,\E,\T)}(\S) \prec \row_{(\S,\E',\T')}(\S)$.
}

\begin{proof}
	For readability, let $\row$ denote $\row_{(\S,\E,\T)}$ and 
	$\row'$ denote $\row_{(\S,\E',\T')}$.
	By the definition of $\prec$, it suffices to show that
	there is an equivariant surjection from a subset of $\rSETd(\S)$ to 
	$\rSET(\S)$
	and they are not isomorphic.
	We first show that
	the function $h:\rSETd(\S) \to \rSET(\S)$ 
	defined by $\rSETd(s) \mapsto \rSET(s)$
	is surjective and equivariant.
	By the following, $h$ is well-defined: 
	\begin{align*}
		\rSETd(s_1) = \rSETd(s_2)
		& \Leftrightarrow \forall e \in \E'.\T'(e[s_1]) = \T'(e[s_2]) \\
		& \Rightarrow \forall e \in \E.\T(e[s_1]) = \T(e[s_2]) \\
		& \Leftrightarrow \rSET(s_1) = \rSET(s_2).
	\end{align*}
	Moreover, $h$ is equivariant by 
	$\rSETd(s) \cdot \pi = \rSETd(s \cdot \pi) \mapsto \rSET(s \cdot \pi) = \rSET(s) \cdot \pi$.
	It is easy to see that $h$ is surjective.

	Next, we show that $\rSET(\S)$ and $\rSETd(\S)$ are not isomorphic.
	To show this, we only have to show that 
	\begin{equation}
		\rSET(\Orbit(s_1))(=\rSET(\Orbit(s_2)) \prec \rSETd(\Orbit(s_1)) \cup \rSETd(\Orbit(s_2))
	\end{equation}
	for $s_1,s_2 \in \S$ such that $\rSET(s_1) = \rSET(s_2)$ and $\rSETd(s_1) \neq \rSETd(s_2)$
	because $\rSET(\S)$ and $\rSETd(\S)$ are orbit finite sets.
	Let $h'$ be $h$ where the domain is restricted to
	$\rSETd(\Orbit(s_1)) \cup \rSETd(\Orbit(s_2))$.
	It is easy to see that $h'$ is a surjective and equivariant function from
	$\rSETd(\Orbit(s_1)) \cup \rSETd(\Orbit(s_2))$ to $\rSET(\Orbit(s_1))$.
	We show that $\rSETd(\Orbit(s_1)) \cup \rSETd(\Orbit(s_2))$ is not isomorphic to
	$\rSET(\Orbit(s_1))$.
	If $\rSETd(\Orbit(s_1)) \neq \rSETd(\Orbit(s_2))$ then this holds trivially.
	Assume that $\rSETd(\Orbit(s_1)) = \rSETd(\Orbit(s_2))$ and
	let 
	\begin{align*}
		\CSec &= \rSET(\Orbit(s_1)), \text{and} \\
		\DTec &= \rSETd(\Orbit(s_1))
	\end{align*}
	by Proposition \ref{prop:CSec}.
	By the definition of $\CSec$ and $\DTec$,
	$\rSETd(s_1) = [\pi_1|_D]_T$ and $\rSETd(s_2) = [\pi_2|_D]_T$ 
	for some $\pi_1,\pi_2 \in G$.
	By Proposition \ref{prop:func}, we can write $h'$ as 
	$h'([\pi|_D]_T) = [\sigma|_C]_S \cdot \pi$ for some $\sigma \in G$.
	Thus,
	\begin{align*}
		h'(\row'(s_1)) = h([\pi_1|_D]_T) = [\sigma|_C]_S \cdot \pi_1 
			= [(\sigma\pi_1)|_C]_S = \rSET(s_1), \\
		h'(\row'(s_2)) = h([\pi_2|_D]_T) = [\sigma|_C]_S \cdot \pi_2 
			= [(\sigma\pi_2)|_C]_S = \rSET(s_2).
	\end{align*}
	Assume that $\CSec$ and $\DTec$ are isomorphic, i.e.,
	there exists an equivariant bijection $f:\CSec \to \DTec$.
	By Proposition \ref{prop:func},
	$f$ can be written as 
	$f([\pi|_C]_S) = [\rho|_D]_T \cdot \pi$ for some $\rho \in G$.
	We have
	\begin{gather*}
		\rSET(s_1) = \rSET(s_2) \\
		f(\rSET(s_1)) = f(\rSET(s_2)) \\
		f([(\sigma \pi_1)|_C]_S) = f([(\sigma \pi_2)|_C]_S) \\
		[\rho|_D]_T \cdot (\sigma \pi_1) = [\rho|_D]_T \cdot (\sigma \pi_2) \\
		[(\rho \sigma \pi_1)|_D]_T = [(\rho \sigma \pi_2)|_D]_T \\
		(\rho \sigma \pi_1)|_D \equiv_T (\rho \sigma \pi_2)|_D \\
		(\rho \sigma \pi_1)|_D (\rho \sigma \pi_2)|_D^{-1} \in T \\
		(\rho \sigma \pi_1 \pi_2^{-1} \sigma^{-1} \rho^{-1})|_D \in T \\
		(\rho \sigma \pi_1 \pi_2^{-1})|_D \in T\rho \sigma.
	\end{gather*}
	Because $f$ is an equivariant and bijective function,
	$\rho|_D S = T \rho|_D$ by the proof of Lemma \ref{lem:isomorphism}.
	Thus, we have $|S| = |T|$.
	By Proposition \ref{prop:func} and $h'$, $\sigma|_C T \subseteq S \sigma|_C$.
	By $\sigma|_C T \subseteq S \sigma|_C$ and $|S| = |T|$, $\sigma|_C T = S \sigma|_C$.
	By $\rho|_D S = T \rho|_D$ and $\sigma|_C T = S \sigma|_C$, 
	$\rho|_D \sigma|_C T = T \rho|_D \sigma|_C$.
	Thus,
	\begin{gather*}
		(\rho \sigma \pi_1 \pi_2^{-1})|_D \in T \rho \sigma = \rho \sigma T \\
		(\pi_1 \pi_2^{-1})|_D \in T \\
		\pi_1|_D \equiv_T \pi_2|_D \\
		[\pi_1|_D]_T = [\pi_2|_D]_T \\
		\rSETd(s_1) = \rSETd(s_2).
	\end{gather*}
	This is a contradiction, and hence $\CSec$ and $\DTec$ are not isomorphic.
	\qed
\end{proof}

\end{document}